\documentclass{article}

\usepackage[english]{babel}
\usepackage[utf8x]{inputenc}
\usepackage[T1]{fontenc}

\usepackage{amsmath,amssymb,amsthm}
\usepackage{mathtools}
\usepackage{graphicx}
\usepackage{hyperref}
\usepackage{algorithm}
\usepackage[noend]{algpseudocode}
\usepackage{todonotes}

\usepackage{tikz-cd}
\usetikzlibrary{arrows.meta,chains,decorations.pathreplacing}

\newcommand{\F}{\mathbb{F}}

\newcommand{\X}{\mathbb{X}}
\newcommand{\Y}{\mathbb{Y}}
\newcommand{\Z}{\mathbb{Z}}

\newtheorem{theorem}{Theorem}

\newtheorem{proposition}[theorem]{Proposition}
\newtheorem{observation}[theorem]{Observation}

\theoremstyle{definition}

\newtheorem{definition}[theorem]{Definition}

\theoremstyle{remark}
\newtheorem{remark}[theorem]{Remark}

\providecommand*\email[1]{\href{mailto:#1}{#1}}

\usepackage[margin = 4 cm]{geometry}

\title{Galois ring isomorphism problem}
\author{Karan Khathuria\thanks{University of Zurich, Switzerland (\email{karan.khathuria@math.uzh.ch})}}

\date{}

\begin{document}
\maketitle
\begin{abstract}
    Recently, Dor\"{o}z et al. (2017) proposed a new hard problem, called the finite field isomorphism problem, and constructed a fully homomorphic encryption scheme based on this problem. 
    In this paper, we generalize the problem to the case of Galois rings, resulting in the \textit{Galois ring isomorphism problem}. The generalization is achieved by lifting the isomorphism between the corresponding residue fields. As a result, this generalization allows us to construct cryptographic primitives over the ring of integers modulo a prime power, instead of a large prime number.
\end{abstract}
\section{Introduction}

Finite fields have been studied extensively due to their numerous applications in different areas of mathematics and computer science, for example, combinatorics, number theory, design theory, coding theory, cryptography, etc. Finite fields are at the base of the theory of finite commutative rings. A direct generalization of finite fields is the Galois rings. The theory of finite fields and Galois rings are parallel and have many similarities. As a result, it becomes natural to extend and study the applications of finite fields to the case of Galois rings. Some of the known applications of Galois rings include: public-key cryptography over residue class rings of integers \cite{gomez1991galois}, multisequence shift register synthesis over Galois rings \cite{armand}, digital signal processing \cite{lin1994rings,krishna1994rings} and algebraic coding theory \cite{z4codes,blake1972codes,spiegel1978codes}.

In a recent paper \cite{do17}, Doroz et al. presented some new cryptographic applications of finite fields, by introducing a new hard problem called the finite field isomorphism (FFI) problem. Let $p$ be a prime, let $\F_p$ be the finite field with $p$ elements, and let $\phi$ be an isomorphism between two distinct extensions of $\F_p$, say $\X := \F_p[x]/(f(x))$ and $\Y := \F_p[y]/(F(y))$. We start from a `short' element $a(x)$ (an element having small coefficients) in $\X$ and consider its image $\phi(a)$ in $\Y$. Informally, the decisional FFI problem asks one to distinguish $\phi(a)$ from a random element in $\Y$, without having the knowledge of the isomorphism $\phi$. Whereas, the computational FFI problem asks one to find the isomorphism $\phi$, given the representation of $\Y$ and some images of short elements. 
In the same paper, the authors developed a fully homomorphic encryption scheme based on the hardness assumption of the FFI problem. Later, in \cite{ho18}, also a signature scheme was proposed relying on this assumption.

In this work, we generalize the FFI problem to Galois rings. In particular, we define the Galois ring isomorphism (GRI) problem, that is analogous to the definition of the FFI problem with $\X$ and $\Y$ being two distinct isomorphic Galois rings. In this case, the coefficients of elements in $\X$ and $\Y$ live in $\Z/p^s\Z$, instead of $\Z/p\Z$. The analysis of the hardness of the FFI problem can simultaneously be extended to study the hardness of the GRI problem. Similar to the case of FFI problem, we describe some lattice based attacks to solve the GRI problem.

With respect to practical applications of the GRI problem, we need efficient algorithms to construct and use isomorphisms between two Galois rings. This can be achieved by lifting the isomorphism between the corresponding residue fields. Galois rings are finite local rings having maximal ideal $(p)$ and residue field $\F_{p^n}$, for some prime integer $p$. Therefore, an isomorphism between Galois rings naturally induces an isomorphism between the corresponding residue fields. Moreover, the isomorphisms between Galois rings are in one-to-one correspondence with the isomorphisms between the respective residue fields.   

The main advantage of this generalization is that the applications of the FFI problem can directly be extended to the Galois rings. Hence, we get a fully homomorphic encryption scheme and a signature scheme over the rings of integers modulo a prime power. This is an advantage because it is generally more efficient to do operations modulo $2^{64}$ than a $64$-bit prime. This is because the arithmetic logic units (ALUs) in the CPU perform arithmetic operations on integer binary numbers, which implies the operations modulo $2^s$ can be performed trivially whereas operations modulo large prime number involve modular reductions. Moreover, using the Chinese remainder theorem, we can further generalize the results over any arbitrary integer modulo ring.

The organization of the paper is as follows. In Section \ref{sec:FFI}, we recall the finite field isomorphism problem and discuss its important aspects, namely hardness and construction of isomorphisms. In Section \ref{sec:GRI}, we define the Galois ring isomorphism problem and we provide an algorithm to construct Galois ring isomorphism by lifting the isomorphisms between the corresponding residue fields. Then, in the same section, we discuss the hardness of the GRI problem and describe some lattice based techniques to solve it. In Section \ref{sec:further_gen}, we remark about further generalization of the problem by using the Chinese remainder theorem. Finally, in Section \ref{sec:conclusion}, we draw some concluding remarks and discuss further works.

\section{Finite field isomorphism problem} \label{sec:FFI}

In \cite{do17}, Dor\"{o}z et al. proposed an encryption scheme based on a new (computationally) hard problem, called the finite field isomorphism problem. It relies on the difficulty of recovering a secret isomorphism between two finite fields. 

Let $\F_p$ be the finite field with $p$ elements, where $p$ is a prime number. Let  $f(x) \in \F_p[x]$ and $F(y) \in \F_p[y]$ be monic irreducible polynomials of degree $n$. Then $\X := \F_p[x]/(f(x))$ and $\Y := \F_p[y]/(F(y))$ are isomorphic fields with $p^n$ elements. Let $\phi$ be an isomorphism from $\X$ to $\Y$.

We will use the variable $x$ and lower case letters for polynomials in $\X$ and the variable $y$ with upper case letters for polynomials in $\Y$. We will perform reductions mod $p$ into the interval $(-p/2,p/2]$.
Let $\chi_\beta$ be a distribution that produces samples from $\X$ having coefficients bounded between $ -\beta$ and $\beta$, given $1 \leq  \beta \leq p/2$.

\begin{definition}[Finite Field Isomorphism Problems (FFI)]
Let $\X, \Y, \phi, \chi_\beta$ be as before and let $a_1(x),\dots,a_k(x)$ be samples from $\chi_\beta$ with corresponding images $A_1(y),\dots,A_k(y)$ in $\Y$.
\begin{itemize}
    \item The \emph{computational} FFI problem (CFFI) is: Given $\Y,A_1(y),\dots,A_k(y)$, recover $f(x)$ and/or the preimages $a_1(x),\dots,a_k(x)$.
    \item The \emph{decisional} FFI problem (DFFI) is: Given $\Y, A_1(y),\dots,A_k(y)$ and $B_1(y),$ $B_2(y)$, where one of $B_1(y), B_2(y)$ is the image of a sample from $\chi_\beta$ and the other one is taken uniformly from $\Y$, distinguish, with a probability greater than 1/2, the element that was constructed using $\phi$.
\end{itemize}
\end{definition}

The hardness of the FFI problem is based on the following experimental observation.

\begin{observation}\cite[Observation 1]{do17}\label{obs1}
Let $f(x)\in\F_p[x]$ and $F(y)\in\F_p[y]$ be chosen uniformly from the set of monic irreducible polynomials of degree $n$. Let $\X ,\Y, \phi$ and $\chi_\beta$ be defined as before. 
Then the image in $\Y$ of a collection of polynomials in $\X$ sampled from $\chi_\beta$ is computationally hard to distinguish from a collection of polynomials sampled uniformly in $\Y$. By a proper choice of parameters, the ability to distinguish such a collection can be made arbitrarily hard.
\end{observation}

In the same paper, Dor\"oz et al. presented three ways to solve the FFI problem. The first two are based on lattice reduction algorithms, whereas the third one is based on solving non-linear polynomial equations. 
\begin{itemize}
    \item \textit{Lattice attacks:} The isomorphism $\phi: \X \to \Y$ is also an $\F_p$-vector space isomorphism. Hence, $\phi^{-1}$ can be described by an $n \times n$ matrix $M$, i.e., if $\phi^{-1}(A(y)) = a(x),$ then $ a = AM$ where $A,a$ are the coefficient vectors of the polynomials $A(y)$ and $a(x)$, respectively. By collecting several samples of images of short polynomials in $\Y$, one gets an instance of a shortest vector problem in a publicly known lattice. In Section \ref{sec:hardness_GRI}, we describe, in detail, an analogue of this attack for the case of the Galois ring isomorphism problem.
    \item \textit{Non-linear algebraic attack:} The strategy here is to recover the image of $x$ (and hence the isomorphism $\phi$) by solving a system of high degree multivariate polynomial equations. Suppose the attacker knows $A_1(y),A_2(y) \in \Y$ such that they are images of the elements  $a_1(x),a_2(x) \in \X$ sampled from $\chi_\beta$. Then for each $i \in \{1,2\}$, $A_i(y) = a_i(\phi(x)) \pmod{F(y)}$. By equating the coefficients, the attacker obtains $2n$ non-linear equations in $3n$ unknowns (the coefficients of $a_1(x),a_2(x)$ and $\phi(x)$). Then the attacker can eliminate the coefficients of $\phi(x)$ (e.g. by using Gr\"{o}bner basis algorithms) to obtain a system of $n$ non-linear equations in $2n$ unknowns. Solving such a system appears to be exponentially difficult. 
\end{itemize}

\subsection{Constructing an Isomorphism} \label{sec:isomorphism}
For any cryptographic application of the FFI problem, we need to efficiently construct an isomorphism between finite fields. This can be achieved using Algorithm \ref{algo:FFI}.

\begin{algorithm}[ht]
\caption{Finite field isomorphism construction}\label{algo:FFI}
\begin{flushleft}
Input: monic irreducible polynomials $f(x),F(y)$ of degree $n$ over $\F_p$.\\ 
Output: an isomorphism $\phi: \X \to \Y$, where $\X := \F_p[x]/(f(x))$ and $\Y := \F_p[y]/(F(y))$.
\end{flushleft}
\begin{algorithmic}[1]
\State  Find a root $\alpha \in \Y$ of $f(t) \in \Y[t]$ and define the map $\phi$ by setting $\phi(x) = \alpha$. Let $A(y)$ be the polynomial representation of $\alpha$ in $\Y$. 
\State Find $\beta \in \X$ such that $F(\beta) = 0 \in \X$ and $A(\beta) = x \pmod{f(x)}$. This defines the inverse map $\phi^{-1}: \Y \to \X$ by setting $\phi^{-1}(y) = \beta$. 
\end{algorithmic}
\end{algorithm}

Algorithm \ref{algo:FFI} consists of two main steps. Let $\X := \F_p[x]/(f(x))$ and $\Y := \F_p[y]/(F(y))$, as in the algorithm. The first step is to compute a root $\alpha \in \Y$ of $f(t) \in \Y[t]$, which can be done using fast polynomial time algorithms, for example, \textsf{polrootsff} routine of Pari-GP \cite{PARI2}. The second step is to compute the inverse isomorphism by finding $\beta \in \X$ which is a root of $F(t) \in \X[t]$ and satisfy $A(\beta) - x = 0$ in $\X$, where $A(y)$ is the polynomial representation of $\alpha$ in $\Y$. For more details on Algorithm \ref{algo:FFI}, we refer the interested reader to \cite[Algorithm 1]{do17}.

 \begin{remark}
 The polynomial $f(x)$ must be chosen independently from the polynomial $F(y)$ such that knowledge of $F(y)$ does not give any information about $f(x)$. Otherwise, the attacker may try to gather information about $f(x)$ when she has isomorphic images in $\Y$ of short polynomials in $\X$.
 \end{remark}

\section{Galois ring isomorphism problem} \label{sec:GRI}
In this section, we generalize the finite field isomorphism problem to the case of Galois rings. We first recall the definition and properties of Galois rings. We refer to \cite{mcdonald1974finite} as a classical reference for the theory of finite rings and to \cite{bini2012finite,wan2003lectures} for more detailed reference for Galois rings. 

\begin{definition}
Let $p$ be a prime number, $s$ and $n$ be positive integers. A Galois ring $GR(p^s,n)$ is a finite commutative unitary local ring of characteristic $p^s$ and cardinality $p^{sn}$.
\end{definition}

\begin{proposition}\cite[Theorem 14.6, Lemma 14.2]{wan2003lectures} Let $GR(p^s,n)$ be a Galois ring. Then 
\begin{enumerate}
    \item $GR(p^s,n)$ is isomorphic to $\left(\Z/p^s\Z\right)[x]/(f(x))$, for any monic polynomial $f(x) \in \left(\Z/p^s\Z\right)[x]$ of degree $n$ whose reduction modulo $p$ is irreducible in $\left(\Z/p\Z\right)[x]$.
    \item $GR(p^s,n)$ is a local ring with a maximal ideal $(p) = p GR(p^s,n)$ and residue field $GR(p^s,n)/(p) \cong \F_{p^n}$. In particular, we have a homomorphism $\pi: GR(p^s,n) \to \F_{p^n}$ given by taking reduction modulo $p$, whose kernel is $(p)$.
\end{enumerate}
\end{proposition}

Let $\X$ and $\Y$ be two (isomorphic) Galois rings of characteristic $p^s$ and cardinality $p^{sn}$. More precisely, let $f(x) \in \left(\Z/p^s\Z\right)[x]$ and $F(y) \in \left(\Z/p^s\Z\right)[y]$ be monic polynomials of degree $n$ such that they are irreducible modulo $p$, and define:
\[ \X = \left(\Z/p^s\Z\right)[x]/(f(x)) \qquad \mbox{and} \qquad \Y = \left(\Z/p^s\Z\right)[y]/(F(y)).\]
Let $\phi$ be an isomorphism between $\X$ and $\Y$. 

All the reductions we perform will be centered at 0, i.e., reductions mod $p$ belong to the interval $(-p/2,p/2]$, and reductions mod $p^s$ belong to the interval $(-p^s/2,p^s/2]$.
Let $\chi_\beta$ be a distribution that produces samples from $\X$ having coefficients bounded between $ -\beta$ and $\beta$, given $1 \leq  \beta < p^s/2$.

\begin{definition}[Galois ring isomorphism problems (GRI)] \label{def:GRI}

Let $\X,\Y,\phi,\chi_\beta$ as before and let $a_1(x),\dots,a_k(x)$ be samples from $\chi_\beta$ with corresponding images $A_1(y),\dots,A_k(y)$ in $\Y$.
\begin{itemize}
    \item The \emph{computational} GRI problem (CGRI) is: Given $\Y,A_1(y),\dots,A_k(y)$, recover $f(x)$ and/or the preimages $a_1(x),\dots,a_k(x)$.
    \item The \emph{decisional} GRI problem (DGRI) is: Given $\Y, A_1(y),\dots,A_k(y)$ and $B_1(y),$ $B_2(y)$, where one of $B_1(y), B_2(y)$ is the image of a sample from $\chi_\beta$ and the other one is taken uniformly from $\Y$, distinguish, with a probability greater than 1/2, the element that was constructed using $\phi$.
\end{itemize}
\end{definition}

\subsection{Algorithm to construct Galois ring isomorphisms}

Let $f,F,\X,\Y$ be described as above. To construct an isomorphism $\phi$ between $\X$ and $\Y$, we lift the isomorphism between their residue fields. In this section, we will use the notation $\overline{\alpha}$ (resp. $\overline{f}(x)$) to denote $\alpha \pmod{p}$ (resp. $f(x) \pmod{p}$) for any $\alpha \in \Z/p^s\Z$ (resp. $f(x) \in \Z/p^s\Z[x]$).

Let $\phi$ be an isomorphism between $\X$ and $\Y$, and let $\pi_\X:\X \to \F_p[x]/\left(\overline{f}(x)\right)$ and $\pi_\Y: \Y \to \F_p[y]/\left(\overline{F}(Y)\right)$ be homomorphims given by reduction modulo $p$. Then we obtain the following commutative diagram:
\[ \begin{tikzcd}
 \X = (\Z/p^s\Z)[x]/(f(x)) \arrow{r}{\phi} \arrow[swap]{d}{\normalsize \pi_\X} & \Y = (\Z/p^s\Z)[y]/(F(y)) \arrow{d}{\pi_\Y}  \\%
 \F_p[x]/\left(\overline{f}(x)\right) \arrow{r}{\overline{\phi}}& \F_p[y]/\left(\overline{F}(y)\right)
\end{tikzcd}
\]
Since $\phi$ is an isomorphism, the induced map $\overline{\phi}$ is a field isomorphism. Conversely, given an isomorphism $\overline{\phi}: \F_p[x]/\left(\overline{f}(x)\right) \to \F_p[y]/\left(\overline{F}(y)\right)$, there exists a unique isomorphism $\phi: \X \to \Y$ that induces $\overline{\phi}$, see \cite[Theorem 5.1.4]{bini2012finite}. The construction and uniqueness of $\phi$ follows from the following proposition, which is a special case of \cite[Lemma 14.5]{wan2003lectures}. We provide the proof because it proves the correctness of Algorithm \ref{algo:GRI} that describes the construction of $\phi$.

\begin{proposition}\cite[Lemma 14.5]{wan2003lectures}
Let $R = GR(p^s,n)$ be a Galois ring and consider the homomorphism $\pi: R \to R/(p)$. Let $f(x)$ be a degree $n$ monic polynomial over $\Z/p^s\Z$ such that $\overline{f}(x)$ is irreducible over $\Z/p\Z$. Further, let $\overline{\alpha}$ be a root of $\overline{f}(x)$ in $R/(p)$. Then there exists a unique root $\beta \in R$ of $f(x)$ such that $\pi(\beta) = \overline{\alpha}$.
\end{proposition}

\begin{proof}
First note that $\overline{f}(x)$ has degree $n$ because $f(x)$ is monic. Thus, $R/(p) \cong \F_{p^n}$ and $\overline{f}(x)$ has $n$ distinct roots in $R/(p)$. This implies that $\overline{\alpha}$ exists and it is a simple root of $\overline{f}(x)$. 

Let $\alpha \in R$ be the trivial lift of $\overline{\alpha}$. We construct a sequence $\beta_0,\beta_1,\ldots,\beta_{s-1} \in R$ as follows:
\begin{align*}
    \beta_0 &= \alpha, \\
    \beta_{i+1} &= \beta_i - \left(f^\prime(\beta_i)\right)^{-1} f(\beta_i) \mbox{ for } i=0,\ldots,s-2.
\end{align*}
We will show that $\beta_{s-1}$ is a root of $f(x)$ and $\pi(\beta_{s-1}) = \overline{\alpha}$.

We first show that $\pi(\beta_i) = \overline{\alpha}$ for $i = 0,1,\ldots,s-1$. Clearly, $\pi(\beta_0) = \pi(\alpha)= \overline{\alpha}$. Now, let $i \geq 1$ and assume $\pi(\beta_{i-1}) = \overline{\alpha}$. Since $\overline{\alpha}$ is a simple root of $\overline{f}(x)$, we have that $\pi(f^\prime(\beta_{i-1})) = \overline{f^\prime}(\overline{\alpha}) = \left(\overline{f} \right)^\prime(\overline{\alpha}) \neq 0$. Hence, $f^\prime(\beta_{i-1})$ is a unit of $R$. Thus, \begin{align*}
    \pi(\beta_i) & = \pi(\beta_{i-1}) - \pi\left(\left(f^\prime(\beta_{i-1})\right)^{-1}\right) \pi(f(\beta_{i-1})) \\
     & = \overline{\alpha} - \pi\left(\left(f^\prime(\beta_{i-1})\right)^{-1}\right) \overline{f}(\overline{\alpha})\\
     & = \overline{\alpha}.
\end{align*} 

Now, we show that $f(\beta_i) \in (p^{i+1})$ for all $i=0,1,\ldots,s-1$. Clearly, since $\pi(f(\beta_0)) = \overline{f}(\overline{\alpha}) = 0$, we get $f(\beta_0) \in (p)$. Now, let $i \geq 1$ and assume that $f(\beta_{i-1}) \in (p^i)$. Then by Taylor's formula\footnote{Let $S$ be a ring and $f(x) = a_nx^n + \cdots + a_0 \in S[x]$. Then the Taylor's formula is given by \begin{align*} f(x+a) &= f(x) +  \frac{f^\prime(x)}{1!} a + \frac{f^{\prime \prime}(x)}{2!}a^2 + \cdots + \frac{f^{(n)}(x)}{n!}a^n, \mbox{ where } \\ 
    f^\prime(x) &= na_n x^{n-1} + (n-1)a_{n-1}x^{n-2} + \cdots + a_1,\\
    f^{(m)}(x) &= \left(f^{(m-1)}(x)\right)^\prime \mbox{ for } m \geq 2.
\end{align*} We note that the Taylor's formula is well-defined as $k!$ divides each coefficient of $f^{(k)}(x)$ for all $k \in \{2,\ldots,n\}$.} we have
\begin{align*}
f(\beta_{i}) &= f\left(\beta_{i-1} - f^\prime(\beta_{i-1})^{-1}f(\beta_{i-1}) \right) \\
& = f(\beta_{i-1}) + \frac{f^\prime(\beta_{i-1})}{1!} \left(-f^\prime(\beta_{i-1})^{-1}f(\beta_{i-1})\right) \\ 
& \qquad +  \frac{f^{\prime\prime}(\beta_{i-1})}{2!}\left(-f^\prime(\beta_{i-1})^{-1}f(\beta_{i-1})\right)^2 + \cdots \\
& \qquad + \frac{f^{(n)}(\beta_{i-1})}{n!}\left(-f^\prime(\beta_{i-1})^{-1}f(\beta_{i-1})\right)^n \\
& = \frac{f^{\prime\prime}(\beta_{i-1})}{2!}\left(-f^\prime(\beta_{i-1})^{-1}f(\beta_{i-1})\right)^2 + \cdots \\
& \qquad + \frac{f^{(n)}(\beta_{i-1})}{n!}\left(-f^\prime(\beta_{i-1})^{-1}f(\beta_{i-1})\right)^n \\
& \in (p^{i+1})
\end{align*} 
Thus, by induction, we have $\beta_{s-1} \in R$ such that $\pi(\beta_{s-1}) = \overline{\alpha}$ and $f(\beta_{s-1}) \in (p^s) =(0)$. Thus we found a root $\beta:= \beta_{s-1} \in R$ of $f(x)$ with $\pi(\beta) = \overline{\alpha}$.
 
To prove uniqueness, we first write $f(x) = (x-\beta)g(x)$, for some $g(x) \in R$. Observe that $\overline{g}(\overline{\alpha}) \neq 0$, as $\overline{\alpha}$ is a simple root of $f(x)$. 
Let $\beta^\prime$ be another root of $f(x)$ in $R$ such that $\pi(\beta^\prime) = \overline{\alpha}$. Then we have \[0 = f(\beta^\prime) = (\beta^\prime - \beta) g(\beta^\prime).\] Moreover, $\pi(g(\beta^\prime)) = \overline{g}(\overline{\alpha}) \neq 0$, which implies $g(\beta^\prime)$ is a unit of $R$ and hence $\beta^\prime = \beta$.
\end{proof}


\begin{algorithm}[ht]
\caption{Galois ring isomorphism construction}\label{algo:GRI}
\begin{flushleft}
Input: $f(x) \in \left(\Z/p^s\Z\right)[x]$ and $F(y) \in \left(\Z/p^s\Z\right)[y]$ monic polynomials of degree $n$ such that they are irreducible modulo $p$.\\ 
Output: an isomorphism $\phi: \X \to \Y$, where $\X := (\Z/p^s\Z)[x]/(f(x))$ and $\Y := (\Z/p^s\Z)[y]/(F(y))$.
\end{flushleft}
\begin{algorithmic}[1]
\State Compute $\overline{f}(x)$ and $\overline{F}(y)$, reductions of $f(x)$ and $F(y)$, respectively, modulo $p$. 
\State Construct an isomorphism $\overline{\phi}: \F_p[x]/\left(\overline{f}(x)\right) \to \F_p[y]/\left(\overline{F}(y)\right)$ using Algorithm \ref{algo:FFI} with input polynomials $\overline{f}(x)$ and $\overline{F}(y)$. Let $\overline{\alpha}$ be the image of $x$ under $\overline{\phi}$, i.e., $\overline{f}\left(\overline{\alpha}\right) = 0 \in \F_p[y]/\left(\overline{F}(y)\right)$.
\State Define $\beta_0 := \alpha \in \Y$, where $\alpha$ is the trivial lift of $\overline{\alpha}$.
\For{$i = 0,\ldots,s-2$}
\State define $\beta_{i+1} := \beta_{i} - \left(f^\prime(\beta_i)\right)^{-1} f(\beta_i) \in \Y$
\EndFor
\State Define the isomorphism $\phi$ by setting $\phi(x) = \beta_{s-1}$.
\end{algorithmic}
\end{algorithm}

\subsection{On the hardness of the Galois ring isomorphism problem} \label{sec:hardness_GRI}

Clearly, the Galois ring isomorphism (GRI) problem is a generalization of the finite field isomophism (FFI) problem. Hence, it follows that the GRI problem is as hard as the FFI problem in the sense that a polynomial time algorithm that solves the GRI problem would also solve the FFI problem. 

In some special cases, the GRI problem canonically reduces to the FFI problem. In particular, reduction to FFI problem is possible when $\beta < p/2$. Let $\X,\Y,\phi,\chi_\beta,$ $ a_1(x),\ldots,a_k(x),$ $A_1(y),\ldots,A_k(y)$ be as in Definition \ref{def:GRI}. By taking reduction modulo $p$, we obtain an instance of the FFI problem. Let $\pi_\X: \X \to \overline{\X}:= \F_p[x]/(\overline{f}(x))$ and $\pi_\Y: \Y \to \overline{\Y}:= \F_p[y]/(\overline{F}(y))$ be the homomorphisms given by taking reduction modulo $p$. Then for all $i \in \{1,\ldots,k\}$ we have $\pi_\X(a_i(x))= \overline{a}_i(x) = a_i(x)$, as the coefficients of $a_i$'s belong to $(-\beta,\beta]$. This gives an instance of the FFI problem with input $\overline{\X},\overline{\Y}, \chi^\prime_\beta$, $\pi_\X(a_1(x)),\ldots,\pi_\X(a_k(x))$, $\pi_\Y(A_1(y)),\ldots,\pi_\Y(A_k(y))$, where $\chi^\prime_\beta$ is a distribution on $\overline{\X}$ that produces samples whose coefficients are bounded between $-\beta$ and $\beta$.

In general if $\beta$ is not less than $p/2$, then the GRI problem does not directly reduces to the FFI problem. However, we can generalize the techniques of solving FFI problem to solve the GRI problem.

\subsubsection{Solving the GRI problem using lattice reduction algorithms}
In \cite{do17}, Dor\"{o}z et al. provided two ways to solve the FFI problem using lattice reduction techniques. In the following, we extend these techniques to solve the GRI problem. 

Let $\X, \Y, \phi \chi_\beta$, $a(x), A(y)$ be as in the Definition \ref{def:GRI}. We first note that $\phi$ (and equivalently $\phi^{-1}$) is also a $\Z/p^s\Z$-module homomorphism, because $\phi$ is given by $\sum_{i=0}^{n-1}a_ix^i \mapsto \sum_{i=0}^{n-1}a_i \phi(x)^i \pmod{F(y)}$. 

We identify a polynomial $a(x)$ with the coefficients vector $a=(a_0,\ldots,a_{n-1})$. Then the map $\phi$ can be described by using an $n \times n$ matrix over $\Z/p^s\Z$.
Let, for each $i \in \{0,\ldots,n-1\}$, $$\phi(x)^i := \sum_{j=0}^{n-1} m_{i,j} y^j \in \Y,$$ for some $m_{i,j} \in \Z/p^s\Z$. Let $M = (m_{i,j})$ be the corresponding matrix. Then, $A(y) = a(\phi(x)) \pmod{F(y)}$ implies that $A = aM$ over $\Z/p^s\Z$.
Moreover, the matrix $M$ is invertible, as there exists an $n \times n$ matrix $N = (n_{i,j}) \in \left(\Z/p^s\Z\right)^{n \times n}$, corresponding to $\phi^{-1}$ with respect to bases $\{1,x,\ldots,x^{n-1}\}$ and $\{1,y,\ldots,y^{n-1}\}$, such that $NM=MN=I_{n}$.

Using this vector/matrix correspondence, we can extend the two lattice attacks presented in \cite{do17} to our case. In the following, we describe one of the attacks in which the corresponding lattice has dimension nearly $2n$. 

Let $A_1,\ldots,A_k$ be the known vectors and $a_1,\ldots,a_k$ be the unknown vectors, corresponding to the polynomials $A_1(y),\ldots,A_k(y)$ and $a_1(x),\ldots,a_k(x)$, respectively. Then the unknown vectors $a_1 = A_1 N,\ldots, a_k = A_k N$ are small, i.e., the absolute value of the entries is bounded by $\beta$. We consider a single coordinate of these vectors, let $N_j^\intercal$ be the $j$-th column of the matrix $N$ and let \[b_j := (A_1 N_j^\intercal, A_2 N_j^\intercal, \ldots, A_k N_j^\intercal).\] Now, define the matrices \[P := (A_1^\intercal| A_2^\intercal | \cdots | A_k^\intercal), \quad Q := (a_1^\intercal|a_2^\intercal|\cdots|a_k^\intercal),\] and set \[D = \begin{pmatrix} P \\ p^s I_k \end{pmatrix}.\]

The matrices $P$ and $Q$ are of dimension $n\times k$ with entries from $\Z/p^s\Z$, the matrix $D$ has dimension $(n+k)\times k$ and $b_j$ is the vector consisting of the $j$-th coordinates of $a_i$'s.

Let $L(D)$ be the lattice generated by the rows of $D$, so it has dimension $k$ and it contains the short row vector $b_j$. If $k$ is chosen sufficiently large, then the vectors $b_j$ will be short relative to the Gaussian heuristic. As a result, such vectors (or a linear combination of them) can be recovered using lattice reduction algorithms.

\section{Notes on further generalizations} \label{sec:further_gen}

It is possible to further generalize the FFI problem in many different ways.
In the most general sense, we can define isomorphism between arbitrary finite commutative rings. Every finite commutative ring is a direct sum of finitely many finite local rings, see \cite[Theorem 3.1.4]{bini2012finite}. Moreover, each of these finite local rings is a homomorphic image of a polynomial ring over a Galois ring, see \cite[Theorem 6.3.1]{bini2012finite}. This implies that an isomorphism between finite rings can be obtained by constructing an isomorphism between the direct summands. 

For more practical purposes, we may restrict ourselves to the case of finite rings that are direct sums of Galois rings. 
Let $R_1,\ldots,R_k$ be Galois rings given by $R_i = \left(\Z/p_i^{s_i}\Z\right)[x]/(f_i(x))$ for each $i \in \{1,\ldots, k\}$, where $p_1,\ldots,p_k$ are distinct primes and for each $i \in \{1,\ldots,k\}$, $f_i(x)$ is a monic polynomial of degree $n$ whose reduction modulo $p_i$ is irreducible in $\left(\Z/p_i^{s_i}\Z\right)[x]$. Let $R = R_1 \oplus \cdots \oplus R_k$, i.e., \[R = \Z[x]/(p_1^{s_1},f_1(x)) \oplus \cdots \oplus \Z[x]/(p_k^{s_k},f_k(x)).\] Since the ideals $(p_i^{s_i},f_i(x))$ and $(p_j^{s_j},f_j(x))$ are co-maximal for all $i \neq j$, we can apply the Chinese remainder theorem to obtain the ring isomorphism \[R \cong \Z[x]/I, \] where $I = \prod_{i=1}^k (p_i^{s_i},f_i(x))$. One can check that $I = (m,f(x))$, where $m = p_1^{s_1} \cdots p_k^{s_k}$ and $f(x)$ is a monic polynomial of degree $n$ such that $f(x) \equiv f_i(x) \pmod{p_i^{s_i}}$ for each $i \in \{1,\ldots,k\}$. The polynomial $f(x)$ is unique modulo $m$, and can be obtained by applying the Chinese remainder theorem (over $\Z/m\Z$) for each coefficient of $f(x)$. This implies that \[R \cong \left(\Z/m\Z\right)/(f(x)),\] for a polynomial $f(x)$ that is irreducible modulo $p_i$ for each $i \in \{1,\ldots,k\}$. Hence, the isomorphisms between such rings can be constructed by constructing isomorphisms between each Galois ring component and then applying the Chinese remainder theorem. Consequently, we can extend the definition of isomorphism problem for such finite rings. Moreover, all the operations to construct an isomorphism are efficient and can be used for practical purposes.

\section{Conclusion} \label{sec:conclusion}

In this paper, we generalize the finite field isomorphism (FFI) problem to Galois rings, and define the Galois ring isomorphism (GRI) problem. We observe that, as in the case of the FFI problem, the best known techniques for solving the GRI problem is based on the lattice reduction algorithms. We show that the construction of a Galois ring isomorphism can efficiently be done by constructively lifting the isomorphism between corresponding residue fields. 

In \cite{do17} and \cite{ho18}, we have seen two applications of the FFI problem, namely, a fully homomorphic encryption scheme and a signature scheme, respectively. As a result of the generalization, we note that both the applications can be extended to the case of Galois rings (or more generally to the case of direct products of Galois rings). Consequently, we obtain the same cryptographic primitives over integer modulo rings. One major advantage of working over integer ring modulo $2^s$ is the efficiency of the cryptographic protocols: performing operations modulo $2^s$ on CPUs are way better than modulo an $s$-bit large prime number. 

\section*{Acknowledgement}
The author would like to thank Gianira Alfarano, Alessandro Neri and Violetta Weger for several useful discussions. 
This work was supported by Forschungskredit of the University of Zurich grant no. FK-19-080.

\bibliographystyle{plain}
\bibliography{biblio}

\end{document}